\documentclass[12pt, letterpaper]{article}
\usepackage{fullpage}
\usepackage{amsmath, amssymb,amsthm}
\usepackage{enumerate}
\usepackage[lined,boxed]{algorithm2e}
\usepackage{hyperref}

\newtheorem{theorem}{Theorem}[section]
\newtheorem{lemma}[theorem]{Lemma}

\newtheorem{corollary}[theorem]{Corollary}

\numberwithin{equation}{section}
 
\newenvironment{definition}[1][Definition]{\begin{trivlist}
\item[\hskip \labelsep {\bfseries #1}]}{\end{trivlist}}

\newenvironment{remark}[1][Remark]{\begin{trivlist}
\item[\hskip \labelsep {\bfseries #1}]}{\end{trivlist}}

\title{Robust Cooperation in the Prisoner's Dilemma: Program Equilibrium via Provability Logic}
\author{Mihaly Barasz, Paul Christiano, Benja Fallenstein,\\ Marcello Herreshoff, Patrick LaVictoire, and Eliezer Yudkowsky}
\begin{document}

\maketitle
\begin{abstract}
We consider the one-shot Prisoner's Dilemma between algorithms with read-access to one anothers' source codes, and we use the modal logic of provability to build agents that can achieve mutual cooperation in a manner that is \emph{robust}, in that cooperation does not require exact equality of the agents' source code, and \emph{unexploitable}, meaning that such an agent never cooperates when its opponent defects. We construct a general framework for such ``modal agents'', and study their properties.
\end{abstract}

\section{Introduction}

Can cooperation in a one-shot Prisoner's Dilemma be justified between rational agents? Rapoport \cite{Rapoport} argued in the 1960s that two agents with mutual knowledge of each others' rationality should be able to mutually cooperate. Howard \cite{Howard} explains the argument thus: \begin{quote} Nonetheless arguments have been made in favour of playing C even in a single play of the PD. The one that interests us relies heavily on the usual assumption that both players are completely rational and know everything there is to know about the situation. (So for instance, Row knows that Column is rational, and Column knows that he knows it, and so on.) It can then be argued by Row that Column is an individual very similar to himself and in the same situation as himself. Hence whatever he eventually decides to do, Column will necessarily do the same (just as two good students given the same sum to calculate will necessarily arrive at the same answer). Hence if Row chooses D, so will Column, and each will get 1. However if Row chooses C, so will Column, and each will then get 2. Hence Row should choose C.\end{quote} Hofstadter \cite{Metamagical} described this line of reasoning as ``superrationality'', and held that knowledge of similar cognitive aptitudes should be enough to establish it, though the latter contention is (to say the least) controversial within decision theory. However, one may consider a stronger assumption: what if each agent has some ability to predict in advance the actions of the other? 
\\
\\ This stronger assumption suggests a convenient logical formalism. In the 1980s, Binmore \cite{Binmore} considered game theory between Turing machines which had access to one anothers' G\"{o}del numbers\footnote{Binmore's analysis, however, eschews cooperation in the Prisoner's Dilemma as irrational!}: \begin{quote} ...a player needs to be able to cope with hypotheses about the reasoning processes of the opponents other than simply that which maintains that they are the same as his own. Any other view risks relegating rational players to the role of the ``unlucky'' Bridge expert who usually loses but explains that his play is ``correct'' and would have led to his winning if only the opponents had played ``correctly''. Crudely, rational behavior should include the capacity to exploit bad play by the opponents. \\ In any case, if Turing machines are used to model the players, it is possible to suppose that the play of a game is prefixed by an exchange of the players' G\"{o}del numbers.\end{quote}
Howard \cite{Howard} and McAfee \cite{McAfee} considered the Prisoner's Dilemma in this context\footnote{One can consider the usual one-shot strategies (always cooperate, always defect) as Turing machines that return the same output regardless of the given input; we denote these algorithms as CooperateBot and DefectBot in order to distinguish them from the outputs Cooperate and Defect.}, and each presented an example of an algorithm which would always return an answer, would cooperate if faced with itself, and would never cooperate when the opponent defected. (The solution discussed in both papers was a program that used quining of the source code to implement the algorithm ``cooperate if and only if the opponent's source code is identical to mine''; we represent it in this paper as Algorithm \ref{cliquebot}, which we call CliqueBot on account of the fact that it cooperates only with the `clique' of agents identical to itself.)
\\
\\ More recently, Tennenholtz \cite{Tennenholtz} reproduced this result in the context of other research on multi-agent systems, noting that CliqueBot can be seen as a Nash equilibrium of the expanded game where two players decide which code to submit to the Prisoner's Dilemma with source-code swap. This context (called ``program equilibrium'') led to several novel game-theoretic results, including folk theorems by Fortnow \cite{Fortnow} and Kalai, Kalai, Lehrer and Samet \cite{KKLS}, an answer by Monderer and Tennenholtz \cite{MT} to the problem of seeking strong equilibria (many-agent Prisoner's Dilemmas in which mutual cooperation can be established in a manner that is safe from \emph{coalitions} of defectors), a Bayesian framework by Peters and Szentes \cite{PS}, and more.
\\
\\ However, these approaches have an undesirable property: they restrict the circle of possible cooperators dramatically---in the most extreme case, only to agents that are syntactically identical! (In a moment, we will see that there are many examples of semantically distinct agents such that one would wish one's program to quickly cooperate with each of them.) Thus mutual cooperation is inherently brittle for CliqueBots, and an ecology of such agents would be akin to an all-out war between incompatible cliques.
\\
\\This problem can be patched somewhat, but not cured, by prescribing a list of agents with whom mutual cooperation is desirable, but this approach is inelegant and requires all of the relevant reasoning to happen outside of the system. We'd like to see agents that can reason on their own somewhat.
\\
\\A natural-seeming strategy involves simulating the other agent to see what they do when given one's own source code. Unfortunately, this leads to an infinite regress when two such agents are pitted against one another.
\\
\\One attempt to put mutual cooperation on more stable footing is the model-checking result of van der Hoek, Witteveen, and Wooldridge \cite{HWW}, which seeks ``fixed points'' of strategies that condition their actions on their opponents' output. However, in many interesting cases there are several fixed points, or none at all, and so this approach does not correspond to an algorithm as we would like.
\\
\\ Since the essence of this problem deals in counterfactuals---e.g. ``what would they do if I did this''---it is worth considering modal logic, which was intended to capture reasoning about counterfactuals, and in particular the G\"{o}del-L\"{o}b modal logic \textbf{GL} with provability as its modality. (See Boolos \cite{Boolos} and Lindstr\"{o}m \cite{Lindstrom} for some good references on \textbf{GL}.) That is, if we consider agents that cooperate if and only if they can prove certain logical formulas, the structure of logical provability gives us a genuine framework for counterfactual reasoning, and in particular a powerful and surprising tool known as L\"{o}b's Theorem \cite{Lob}:
\begin{theorem} [L\"{o}b's Theorem] \label{lobtheorem}
Let \textsf{S} be a formal system which includes Peano Arithmetic. If $\phi$ is any well-formed formula in \textsf{S}, let $\Box \phi$ be the formula in a G\"{o}del encoding of \textsf{S} which claims that there exists a proof of $\phi$ in \textsf{S}; then whenever $\textsf{S}\vdash (\Box \phi \to\phi)$, in fact $\textsf{S}\vdash \phi$.
\end{theorem}

We shall see that L\"{o}b's Theorem enables a flexible and secure form of mutual cooperation in this context. In particular, we first consider the intuitively appealing strategy ``cooperate if and only if I can prove that my opponent cooperates'', which we call FairBot. If we trust the formal system used by FairBot, we can conclude that it is unexploitable (in the sense that it never winds up with the sucker's payoff). When we play FairBot against itself (and give both agents sufficient power to find proofs), although either mutual cooperation or mutual defection seem philosophically consistent\footnote{As we shall see, the symmetry between mutual cooperation and mutual defection is broken by the positive criterion for action.}, it always finds mutual cooperation (Theorem \ref{FBFB})!\footnote{This result was proved by Vladimir Slepnev in an unpublished draft \cite{Slepnev}, and the proof is reproduced later in this paper with his permission.} Furthermore, we can construct another agent after the same fashion which improves on the main deficit of the above strategy: namely, that FairBot fails to correctly defect against CooperateBot.\footnote{For a philosophical discussion of why this is the obviously correct response to CooperateBot, see Section \ref{PhilosophicalConclusions}.} We call this agent PrudentBot.
\\
\\ Moreover, the underpinnings of this result (and the others in this paper) do not depend on the syntactical details of the programs, but only on their semantic interpretations in provability logic; therefore two such programs can cooperate, even if written differently (in several senses, for instance if they use different G\"{o}del encodings or different formal systems). Accordingly, we define a certain class of algorithms, such as FairBot and PrudentBot, whose behavior can be described in terms of a modal formula, and show that the actions these ``modal agents'' take against one another can be described purely in terms of these modal formulas. Using the properties of Kripke semantics, one can algorithmically derive the fixed-point solutions to the action of one modal agent against another; indeed, the results of this paper have additionally been checked by a computer program written by two of the authors, hosted at \url{github.com/machine-intelligence/provability}.
\\
\\ We next turn to the question of whether a meaningful sense of optimality exists among modal agents. Alas, there are several distinct obstacles to some natural attempts at a nontrivial and non-vacuous criterion for optimality among modal agents. This echoes the impossibility-of-optimality results of Anderlini \cite{Anderlini} and Canning \cite{Canning} on game theory for Turing machines with access to each others' source codes.
\\
\\ All the same, the results on L\"{o}bian cooperation represent a formalized version of robust mutual cooperation on the Prisoner's Dilemma, further validating some of the intuitions on ``superrationality'' and raising new questions on decision theory. The Prisoner's Dilemma with exchange of source code is analogous to Newcomb's problem, and indeed, this work was inspired by some of the philosophical alternatives to causal and evidential decision theory proposed for that problem (see Drescher \cite{Drescher} and Altair \cite{Altair}).
\\
\\ A brief outline of the structure of this paper: in Section \ref{formal}, we define our formal framework more explicitly. In Section \ref{playfair}, we introduce FairBot, prove that it achieves mutual cooperation with itself and cannot be exploited (Theorem \ref{FBFB}); we then introduce PrudentBot, and show that it is also unexploitable, cooperates mutually with itself and with FairBot, and defects against CooperateBot.
\\
\\ In Section \ref{Modal Agents}, we develop the theory of modal agents, with a focus on showing that their action against one another is well-defined. We also show that a feature of PrudentBot---namely, that it checks its opponent's response against DefectBot---is essential to its functioning: modal agents which do not use third parties cannot achieve mutual cooperation with FairBot unless they also cooperate with CooperateBot. Then, in Section \ref{optimality}, we discuss several obstacles to proving nontrivial optimality results for modal agents. In Section \ref{PhilosophicalConclusions}, we will explain our preference for PrudentBot over FairBot, and speculate on some future directions, before closing in Section \ref{OpenProblems} with a list of open problems in this area.

\section{Agents in Formal Logic}
\label{formal}
There are two different formalisms which we will bear in mind throughout this paper. The first formalism is that of algorithms, where we can imagine two Turing machines \texttt{X} and \texttt{Y}, each of which is given as input the code for the other, and which have clearly defined outputs corresponding to the options $C$ and $D$. (It is possible, of course, that one or both may fail to halt, though the algorithms that we will discuss will provably halt on all inputs.) This formalism has the benefit of concreteness: we could actually program such agents, although the ones we shall deal with are often very far from efficient in their requirements. It has the drawback, however, that proofs about algorithms which call upon each other are generally difficult and untidy, relying upon delicate bounds on (e.g.) the length of proofs.
\\
\\ Therefore, we will do our proofs in another framework: that of logical provability in certain formal systems. More specifically, the agents we will be most interested in can be interpreted via modal formulas in G\"{o}del-L\"{o}b provability logic, which is especially pleasant to work with. 
\\
\\ We believe this bait-and-switch is justified, because we strongly suspect that all of our tools will indeed have equivalently useful bounded versions.  Variants of L\"{o}b's theorem for bounded proof lengths seem readily derivable via currently available techniques in logic, such as those in \cite{Pudlak}.  We therefore encourage readers to construct bounded algorithmic versions of all logically defined agents in this paper, and explore their behavior.  We are confident that with the right bounds and parameters in place, all of our theorems will have analogues for bounded agents. \footnote{As of 2019, some progress in this direction has been made by \cite{Critch}, exhibiting a version of FairBot using bounded proof lengths.}
\\
\\ In particular, our ``agents'' will be formulas in Peano Arithmetic, and our criterion for action will be the existence of a finite proof in the tower of formal systems $\textsf{PA+n}$, where $\textsf{PA}$ is Peano Arithmetic, and $\textsf{PA+(n+1)}$ is the formal system whose axioms are the axioms of $\textsf{PA+n}$, plus the axiom that \textsf{PA+n} is consistent, i.e. that $\neg\Box\dots\Box\bot$ with $n+1$ copies of $\Box$.
\\
\\ Fix a particular G\"{o}del numbering scheme, and let \texttt{X} and \texttt{Y} each denote well-formed formulas with one free variable. Then let $X(Y)$ denote the formula where we replace each instance of the free variable in \texttt{X} with the G\"{o}del number of \texttt{Y}. If such a formula holds in the standard model of Peano Arithmetic, we interpret that as \texttt{X} cooperating with \texttt{Y}; if its negation holds, we interpret that as \texttt{X} defecting against \texttt{Y}. (In particular, we will prove theorems in \textsf{PA+n} to establish whether the agents we discuss cooperate or defect against one another.) Thus we can regard such formulas of arithmetic as decision-theoretic agents, and we will use ``source code'' to refer to their G\"{o}del numbers.
\begin{remark}
To maximize readability in the technical sections of this paper, we will use typewriter font for agents, which are formulas of Peano Arithmetic with a single free variable, like \texttt{X} and \texttt{CooperateBot}; we will use sans-serif font for the formal systems \textsf{PA+n}; and we will use italics for logical formulas with no free variables such as $X(Y)$. Furthermore, we will use $[X(Y)=C]$ and $[X(Y)=D]$ interchangeably with $X(Y)$ and $\neg X(Y)$.
\end{remark}
Of course, it is easy to create \texttt{X} and \texttt{Y} so that $X(Y)$ is an undecidable statement in all \textsf{PA+n} (e.g. the statement that the formal system $\textsf{PA+}\omega$ is consistent). But the philosophical phenomenon we're interested in can be achieved by agents which do not present this problem, and whose finitary versions in fact always return an answer in finite time.\\
\\
Two agents which are easy to define and clearly decidable are the agent which always cooperates (which we will call \texttt{CooperateBot}, or \texttt{CB} for short) and the agent which always defects (which we will call \texttt{DefectBot}, or \texttt{DB}). In pseudocode:

\begin{algorithm}[H]
 \SetKwInOut{Input}{Input}\SetKwInOut{Output}{Output}

 \Input{Source code of the agent \texttt{X}}
 \Output{\emph{C} or \emph{D}}

 return \emph{C}\;
 \caption{\texttt{CooperateBot (CB)}}
\end{algorithm}

\begin{algorithm}[H]
 \SetKwInOut{Input}{Input}\SetKwInOut{Output}{Output}

 \Input{Source code of the agent \texttt{X}}
 \Output{$C$ or $D$}

 return $D$\;
 \caption{\texttt{DefectBot (DB)}}
\end{algorithm}

\begin{remark}
In the Peano Arithmetic formalism, \texttt{CooperateBot} can be represented by a formula that is a tautology for every input, while \texttt{DefectBot} can be represented by the negation of such a formula. For any $\texttt X$, then, \textsf{PA} $\vdash [CB(X)=C]$ and \textsf{PA} $\vdash [DB(X)=D]$.\\
\\
Note further that $\textsf{PA}\not\vdash \neg\Box[DB(X)=C]$, but that $\textsf{PA+1}\vdash \neg\Box [DB(X)=C]$; this distinction is essential.
\end{remark}

Howard \cite{Howard}, McAfee \cite{McAfee} and Tennenholtz \cite{Tennenholtz} introduced functionally equivalent agent schemas, which we've taken to calling \texttt{CliqueBot}; these agents use quining to recognize self-copies and mutually cooperate, while defecting against any other agent. In pseudocode:

\begin{algorithm}[H]
\label{cliquebot}
 \SetKwInOut{Input}{Input}\SetKwInOut{Output}{Output}

 \Input{Source code of the agent \texttt{X}}
 \Output{$C$ or $D$}

\eIf {\emph{\texttt{X}$=$\texttt{CliqueBot}}}{
  return $C$\;
  }{
  return $D$\;
  }
 \caption{\texttt{CliqueBot}}
\end{algorithm}

By the diagonal lemma, there exists a formula of Peano Arithmetic which implements \texttt{CliqueBot}. (The analogous tool for computable functions is Kleene's recursion theorem~\cite{Kleene}; in this paper, we informally use ``quining'' to refer to both of these techniques.)\\
\\
\texttt{CliqueBot} has the nice property that it never experiences the sucker's payoff in the Prisoner's Dilemma. This is such a clearly important property that we will give it a name:
\begin{definition}
We say that an agent \texttt{X} is \emph{unexploitable} if there is no agent \texttt{Y} such that $X(Y)=C$ and $Y(X)=D$.
\end{definition}
However, \texttt{CliqueBot} has a notable drawback: it can only elicit mutual cooperation from agents that are syntactically identical to itself. (If two \texttt{CliqueBots} were written with different G\"{o}del-numbering schemes, for instance, they would defect against one another!)
\\
\\ One might patch this by including a list of source codes (or a schema for them), and cooperate if the opponent matches any of them; one would of course be careful to choose only source codes that would cooperate back with this variant. But this is a brittle form of mutual cooperation, and an opaque one: it takes a predefined circle of mutual cooperators as given. For this reason, it is worth looking for a more flexibly cooperative form of agent, one that can deduce for itself whether another agent is worth cooperating with.

\section{L\"{o}bian Cooperation}
\label{playfair}
A deceptively simple-seeming such agent is one we call \texttt{FairBot}. On a philosophical level, it cooperates with any agent that can be proven to cooperate with it. In pseudocode:

\begin{algorithm}[H]
\label{FairBot}
 \SetKwInOut{Input}{Input}\SetKwInOut{Output}{Output}

 \Input{Source code of the agent \texttt{X}}
 \Output{$C$ or $D$}

 \eIf {\emph{\textsf{PA}} $\vdash$ [X(FairBot)$=C$]}{
  return $C$\;
  }{
  return $D$\;
  }
 \caption{\texttt{FairBot (FB)}}
\end{algorithm}

\texttt{FairBot} references itself in its definition, but as with \texttt{CliqueBot}, this can be done via the diagonal lemma. By inspection, we see that \texttt{FairBot} is unexploitable: presuming that Peano Arithmetic is sound, \texttt{FairBot} will not cooperate with any agent that defects against \texttt{FairBot}.
\\
\\ The interesting question is what happens when \texttt{FairBot} plays against itself: it intuitively seems plausible either that it would mutually cooperate or mutually defect. As it turns out, though, L\"{o}b's Theorem guarantees that since the \texttt{FairBot}s are each seeking proof of mutual cooperation, they both succeed and indeed cooperate with one another! (This was first shown by Vladimir Slepnev \cite{Slepnev}.)
\begin{theorem} \label{FBFB}
\emph{\textsf{PA}} $\vdash [$FairBot(FairBot)$=C]$.
\end{theorem}
\begin{proof}[Proof (Simple Version):]
By inspection of \texttt{FairBot}, \textsf{PA}$\vdash(\Box[FB(FB)=C])\to [FB(FB)=C]$. By L\"{o}b's Theorem, Peano Arithmetic does indeed prove that \emph{FairBot(FairBot)=C}.
\end{proof}

However, it is a tidy logical accident that the two agents are the same; we will understand better the mechanics of mutual cooperation if we pretend in this case that we have two distinct implementations, $\texttt{FairBot}_1$ and $\texttt{FairBot}_2$, and prove mutual cooperation from their formulas without using the fact that their actions are identical.
\begin{proof}[Proof of Theorem \ref{FBFB} (Real Version):]
Let $A$ be the formula ``$FB_1(FB_2)=C$'' and $B$ be the formula ``$FB_2(FB_1)=C$''. By inspection, \textsf{PA}$\vdash \Box A \to B$ and \textsf{PA}$\vdash \Box B \to A$. This sort of ``L\"{o}bian circle'' works out as follows:
\begin{eqnarray*}
&\textsf{PA} \vdash (\Box A\to B)\wedge (\Box B \to A) & \text{(see above)}\\
&\textsf{PA} \vdash (\Box A\wedge \Box B)\to (A \wedge B) & \text{(follows from above)}\\
&\textsf{PA} \vdash \Box(A\wedge B)\to (\Box A \wedge \Box B) & \text{(tautology)}\\
&\textsf{PA} \vdash \Box(A \wedge B) \to (A \wedge B) & \text{(previous lines)}\\
&\textsf{PA} \vdash A\wedge B & \text{(L\"{o}b's Theorem).}
\end{eqnarray*}
\end{proof}
\begin{remark}
One way to build a finitary version of \texttt{FairBot} is to write an agent \texttt{FiniteFairBot} that looks through all proofs of length $\leq N$ to see if any are a proof of $[X(FiniteFairBot)=C]$, and cooperates iff it finds such a proof. If $N$ is large enough, the bounded version of L\"{o}b's Theorem implies the equivalent of Theorem \ref{FBFB}.
\end{remark}
\begin{remark}
Unlike a \texttt{CliqueBot}, \texttt{FairBot} will find mutual cooperation even with versions of itself that are written in other programming languages. In fact, even the choice of formal system does not have to be identical for two versions of \texttt{FairBot} to achieve mutual cooperation! It is enough that there exist a formal system \textsf{S} in which L\"{o}bian statements are true, such that anything provable in \textsf{S} is provable in each of the formal systems used, and such that \textsf{S} can prove the above. (Note in particular that even \emph{incompatible} formal systems can have this property: a version of \texttt{FairBot} which looks for proofs in the formal system \textsf{PA}$+\neg$\textsf{Con(PA)} will still find mutual cooperation with a \texttt{FairBot} that looks for proofs in \textsf{PA+1}.)
\end{remark}
However, \texttt{FairBot} wastes utility by cooperating even with \texttt{CooperateBot}. (See Section \ref{PhilosophicalConclusions} for the reasons we take this as a serious issue.) Thus we would like to find a similarly robust agent which cooperates mutually with itself and with \texttt{FairBot} but which defects against \texttt{CooperateBot}.\\
\\
Consider the agent \texttt{PrudentBot}, defined as follows:

\begin{algorithm}[H]
\label{PrudentBot}
 \SetKwInOut{Input}{Input}\SetKwInOut{Output}{Output}

 \Input{Source code of the agent \texttt{X}}
 \Output{$C$ or $D$}
 \If {\textsf{\emph{PA}} $\vdash$ [X(PrudentBot)=C] \emph{and} \textsf{\emph{PA+1}} $\vdash$ [X(DefectBot)=D]}{
  \Return $C$\;
  }
  return $D$\;
 \caption{\texttt{PrudentBot (PB)}}
\end{algorithm}
\begin{theorem}
\label{PrudentRocks}
\texttt{PrudentBot} is unexploitable, mutually cooperates with itself and with \texttt{FairBot}, and defects against \texttt{CooperateBot}.
\end{theorem}
\begin{proof}
Unexploitability is again immediate from the definition of \texttt{PrudentBot} and the assumption that \textsf{PA} is sound, since cooperation by \texttt{PrudentBot} requires a proof that its opponent cooperates against it.\\
\\
In particular, $\textsf{PA+1}\vdash [PB(DB)=D]$ (since $\textsf{PA}\vdash [DB(PB)=D]$, $\textsf{PA+1}\vdash \neg \Box[DB(PB)=C]$).\\
\\ 
It is likewise clear that $\textsf{PA+2}\vdash [PB(CB)=D]$.\\
\\
Now since $\textsf{PA+1}\vdash [FB(DB)=D]$ and therefore $\textsf{PA}\vdash\Box(\neg\Box\bot\to[FB(DB)=D])$, we again have the L\"{o}bian cycle where $\textsf{PA}\vdash[PB(FB)=C]\leftrightarrow \Box[FB(PB)=C]$, and of course vice versa; thus \texttt{PrudentBot} and \texttt{FairBot} mutually cooperate.
\\
\\ And as we have established $\textsf{PA+1}\vdash [PB(DB)=D]$, we have the same L\"{o}bian cycle for \texttt{PrudentBot} and itself.
\end{proof}

\begin{remark}
It is important that we look for proofs of $X(DB)=D$ in a stronger formal system than we use for proving $X(PB)=C$; if we do otherwise, the resulting variant of \texttt{PrudentBot} would lose the ability to cooperate with itself. However, it is not necessary that the formal system used for $X(DB)=D$ be stronger by only one step than that used for $X(PB)=C$; if we use a much higher $\textsf{PA+n}$ there, we broaden the circle of potential cooperators without thereby sacrificing safety.
\end{remark}

\section{Modal Agents}
\label{Modal Agents}
It is instructive to consider \texttt{FairBot} and \texttt{PrudentBot} as modal statements in G\"{o}del-L\"{o}b provability logic (often denoted \textbf{GL}). Namely, if we consider the actions of \texttt{FairBot} and any other agent \texttt{X} against one another, then the definition of \texttt{FairBot} is simply $[FB(X)=C]\leftrightarrow \Box [X(FB)=C]$, and the definition of \texttt{PrudentBot} is $[PB(X)]\leftrightarrow (\Box[X(PB)]\wedge \Box(\neg\Box\bot\to\neg[X(DB)])).$ There are a number of tools, including fixed-point theorems and Kripke semantics, which work for families of such modal statements; and thus we will define a class of \emph{modal agents} for the purpose of study.
\\
\\ Informally, a modal agent is one whose actions are determined solely by the provability of statements regarding its opponent's actions against itself and against other simpler agents\footnote{For instance, \texttt{PrudentBot} tries to prove that its opponent defects against \texttt{DefectBot}. As for the requirement that these secondary agents be ``simpler'', we want to avoid the possibility of infinite regress using something like Russell's theory of types.}. That is, if \texttt{X} is a modal agent, then there is a modal-logic formula\footnote{That is, a formula built from $\Box$, $\top$, the logical operators $\wedge$, $\vee$, $\rightarrow$, $\leftrightarrow$ and $\neg$, and the input variables.} $\varphi$ and a fixed set of simpler modal agents $\texttt Y_1,\dots,\texttt Y_N$ such that, for any opponent $\texttt Z$, 
\begin{equation}
\label{modalagent}
[X(Z)=C]\leftrightarrow \varphi\left([Z(X)=C],[Z(Y_1)=C],\dots,[Z(Y_N)=C]\right).
\end{equation}
Furthermore, since a modal agent does all of this via provability, the formula $\varphi$ must be \emph{fully modalized}: all instances of variables must be contained inside sub-formulas of the form $\Box\psi$.
\\
\\We must lay some groundwork (following Lindstr\"om \cite{Lindstrom}) before formally defining the class of modal agents. Write $\varphi(p_1,\dotsc,p_n)$ to denote a formula $\varphi$ in the language of \textsf{GL} whose free (propositional) variables are included in the set $\{p_1,\dotsc,p_n\}$. (Note that this is different from Lindstr\"om's convention, who doesn't display the free variables.)

\begin{theorem}[Arithmetic soundness of \textsf{GL}] \label{theorem-arithmetic-soundness}
Suppose that $\mathsf{GL}\vdash\varphi(p_1,\dotsc,p_n)$, and that $\psi_1,\dotsc,\psi_n$ are closed formulas in the language of $\mathsf{PA}$. Then $\mathsf{PA}\vdash\varphi(\psi_1,\dotsc,\psi_n)$.
\end{theorem}

\begin{proof}
This is Theorem 1 of Lindstr\"om.
\end{proof}

\begin{theorem}[Modal fixed point theorem] \label{theorem-fpt}
Suppose that the formula $\varphi(p,q_1,\dotsc,q_n)$ in the language of $\mathsf{GL}$ is modalized in $p$. Then there is a modal formula $\psi(q_1,\dotsc,q_n)$ such that
\[
\mathsf{GL} \;\vdash\; \psi(q_1,\dotsc,q_n) \;\leftrightarrow\; \varphi(\psi(q_1,\dotsc,q_n),\,q_1,\dotsc,q_n).
\]
Moreover, if $\varphi$ is modalized in $q_i$, then so is $\psi$.
\end{theorem}

\begin{proof}
Except for the last sentence, this is Theorem 11 of Lindstr\"om. The last sentence is obvious from Lindstr\"om's proof, since the $\psi$ constructed there differs from $\varphi$ only inside boxes.
\end{proof}

Write $\Box^+\varphi$ to mean $\varphi\wedge\Box\varphi$.

\begin{theorem}[Uniqueness of modal fixed points] \label{theorem-ufpt}
Suppose that $\varphi(p,q_1,\dotsc,q_n)$ is modalized in $p$. Then
\[
\mathsf{GL} \;\vdash\; \Box^+\big(p\leftrightarrow\varphi(p,q_1,\dotsc,q_n)\big) \;\wedge\; \Box^+\big(p'\leftrightarrow\varphi(p',q_1,\dotsc,q_n)\big) \;\to\; (p\leftrightarrow p').
\]
\end{theorem}

\begin{proof}
This is Theorem 12 of Lindstr\"om.
\end{proof}

\begin{corollary}[Uniqueness of arithmetic fixed points] \label{corollary-arithmetic-ufpt}
Suppose that $\varphi(p,q_1,\dotsc,q_n)$ is modalized in $p$, and let $\psi,\psi',\psi_1,\dotsc,\psi_n$ be closed formulas in the language of $\mathsf{PA}$. If
\[
\mathsf{PA}\;\vdash\;\psi\,\leftrightarrow\,\varphi(\psi,\psi_1,\dotsc,\psi_n) \;\quad\;\text{and}\;\quad\; \mathsf{PA}\;\vdash\;\psi'\,\leftrightarrow\,\varphi(\psi',\psi_1,\dotsc,\psi_n),
\]
then $\mathsf{PA}\vdash\psi\leftrightarrow\psi'$.
\end{corollary}

\begin{proof}
By applying Theorem~\ref{theorem-arithmetic-soundness} to the conclusion of Theorem~\ref{theorem-ufpt}, we obtain
\[
\mathsf{PA} \;\vdash\; \Box^+\big(\psi\leftrightarrow\varphi(\psi,\psi_1,\dotsc,\psi_n)\big) \;\wedge\; \Box^+\big(\psi'\leftrightarrow\varphi(\psi',\psi_1,\dotsc,\psi_n)\big) \;\to\; (\psi\leftrightarrow \psi').
\]
But for any formula $\tilde\varphi$, if $\mathsf{PA}\vdash\tilde\varphi$ then $\mathsf{PA}\vdash\Box^+(\tilde\varphi)$, so the conclusion follows.
\end{proof}

\begin{lemma} \label{lemma-modal-substitution}
If $\varphi(p_1,\dotsc,p_n)$ is a modal formula and $\psi_1,\psi_1',\dotsc,\psi_n,\psi_n'$ are arithmetic formulas such that $\mathsf{PA}\vdash\psi_i\leftrightarrow\psi_i'$ for each $i$, then $\mathsf{PA}\vdash\varphi(\psi_1,\dotsc,\psi_n)\leftrightarrow\varphi(\psi_1',\dotsc,\psi_n')$.
\end{lemma}

\begin{proof}
Lindstr\"om's Lemma~8 states that for any modal formula $\tilde\varphi(p,q_1,\dotsc,q_m)$,
\[
\mathsf{GL}\;\vdash\;\Box^+(p\leftrightarrow p')\;\to\;\big(\tilde\varphi(p,q_1,\dotsc,q_m)\leftrightarrow\tilde\varphi(p',q_1,\dotsc,q_m)\big).
\]
The desired result is obtained by applying this $n$ times and then appealing to Theorem~\ref{theorem-arithmetic-soundness}.
\end{proof}

\begin{definition}
An \emph{agent} is a well-formed formula in the language of $\mathsf{PA}$ of (at most) one free variable.
\end{definition}
If $\texttt X$ and $\texttt Y$ are agents, we write $[X(Y)=C]$, or $[X(Y)]$ for short, for the application of $\texttt X$ to the G\"odel number of $\texttt Y$; we interpret this logical formula as the assertion that $\texttt X$ cooperates when playing against~$\texttt Y$. Accordingly, we say that \texttt{X} cooperates with \texttt{Y} if $[X(Y)]$ holds in the standard model of \textsf{PA}.

\begin{definition}
An agent $\texttt X$ is called a \emph{modal agent of rank $k\geq 0$} if there are modal agents $\texttt Y_1,\dotsc,\texttt Y_n$ of rank $<k$ and a fully modalized formula $\varphi(p,q_1,\dotsc,q_n)$ such that such that for all agents~$\texttt Z$,
\[
\mathsf{PA}\;\vdash\;[X(Z)]\;\leftrightarrow\;\varphi([Z(X)],[Z(Y_1)],\dotsc,[Z(Y_n)]).
\]
\end{definition}
\texttt{CooperateBot}, \texttt{DefectBot}, \texttt{FairBot} and \texttt{PrudentBot} are all modal agents, but as we shall see, \texttt{CliqueBot} is not.\\
\\
We now prove three theorems demonstrating that the notion of ``modal agent'' has good properties. First, we note that it makes no practical difference if we include proofs about the actions $[X(Z)]$ and $[Y_i(Z)]$ in our definition: 
\begin{theorem} \label{theorem-self-referential-modal-agents}
Suppose that $\texttt X$ is an agent, $\texttt Y_1,\dotsc,\texttt Y_n$ are modal agents of rank $<k$, and $\varphi(\tilde p,p,\tilde q_1,q_1,\dotsc,\tilde q_n,q_n)$ is a fully modalized formula such that for all agents $\texttt Z$,
\[
\mathsf{PA}\;\vdash\;[X(Z)]\;\leftrightarrow\;\varphi([X(Z)],[Z(X)],[Y_1(Z)],[Z(Y_1)],\dotsc,[Y_n(Z)],[Z(Y_n)]).
\]
Then $\texttt X$ is a rank-$k$ modal agent.
\end{theorem}

\begin{proof}
Without loss of generality, assume that if the modal formula for $\texttt Y_i$ depends on a lower-rank modal agent $\texttt Y'$, then $\texttt Y' = \texttt Y_j$ for some $j$. Then for every~$i$, there is a fully modalized formula $\varphi_i$ such that for all $\texttt Z$, $\mathsf{PA}\vdash[Y_i(Z)]\leftrightarrow\varphi_i([Z(Y_1)],\dots,[Z(Y_n)])$. Thus, by Lemma~\ref{lemma-modal-substitution}, we may assume that $\varphi$ is of the form $\varphi([X(Z)],[Z(X)],[Z(Y_1)],\dotsc,[Z(Y_n)])$. It remains to eliminate the dependency on $[X(Z)]$.

By Theorem~\ref{theorem-fpt}, there is a fully modalized formula $\psi(p,q_1,\dotsc,q_n)$ such that $\mathsf{GL}$ proves $\psi(p,q_1,\dotsc,q_n)\leftrightarrow\varphi(\psi(p,q_1,\dotsc,q_n),\,p,\,q_1,\dotsc,q_n)$. Hence by Theorem~\ref{theorem-arithmetic-soundness}, for all agents $\texttt Z$
\begin{align*}
\mathsf{PA}\;\;\vdash\;\;
& \psi([Z(X)],[Z(Y_1)],\dotsc,[Z(Y_n)])\;\leftrightarrow\; \\
& \qquad\varphi\big(\psi([Z(X)],[Z(Y_1)],\dotsc,[Z(Y_n)]),\;[Z(X)],\;[Z(Y_1)],\dotsc,[Z(Y_n)]\big).
\end{align*}
Since $[X(Z)]$ is also a fixed point of this equation, the conclusion follows by Corollary~\ref{corollary-arithmetic-ufpt}.
\end{proof}

Next, we show that the actions two modal agents take against each other are described by a fixed point of their modal formulas, as one would expect. In particular, this shows that modal agents' actions against each other depend only on their modal formulas, not on other features of their source code.

\begin{theorem}
If $\texttt X$ and $\texttt Y$ are modal agents, define a closed modal formula $\psi_{[X(Y)]}$ by recursion on their rank, as follows. Write $\texttt X_i$ and $\texttt Y_i$ for the lower-rank agents $\texttt X$ and $\texttt Y$ depend on, respectively, and write $\varphi_X(p,q_1,\dotsc,q_m)$ and $\varphi_Y(p,q_1,\dotsc,q_n)$ for the modal formulas corresponding to $\texttt X$ and $\texttt Y$. Now let $\psi_{[X(Y)]}$ be the formula provided by Theorem~\ref{theorem-fpt} satisfying
\[
\mathsf{GL}\;\vdash\;\psi_{[X(Y)]}\;\leftrightarrow\;\varphi_X\big(\varphi_Y(\psi_{[X(Y)]},\,\psi_{[X(Y_1)]},\dotsc,\psi_{[X(Y_n)]}),\;\psi_{[Y(X_1)]},\dotsc,\psi_{[Y(X_m)]}\big).
\]
Then
\[
\mathsf{PA} \;\vdash\; [X(Y)]\;\leftrightarrow\;\psi_{[X(Y)]}.
\]
\end{theorem}

\begin{proof}
By induction, we assume that this already holds for lower ranks. By the definition of modal agent, the induction hypothesis, and Lemma~\ref{lemma-modal-substitution}, $\mathsf{PA}$ shows that $[X(Y)]$ is equivalent to
\[
\varphi_X\big([Y(X)],\psi_{[Y(X_1)]},\dotsc,\psi_{[Y(X_m)]}\big)
\]
and that this is in turn equivalent to
\[
\varphi_X\big(\varphi_Y\big([X(Y)],\psi_{[X(Y_1)]},\dotsc,\psi_{[X(Y_n)]}\big),\;\psi_{[Y(X_1)]},\dotsc,\psi_{[Y(X_m)]}\big).
\]
Thus, $\mathsf{PA}$ shows that both $[X(Y)]$ and $\psi_{[X(Y)]}$ are fixed points of the same formula, and hence that $[X(Y)]\leftrightarrow\psi_{[X(Y)]}$ by Corollary~\ref{corollary-arithmetic-ufpt}.
\end{proof}

Finally, we show that modal agents' actions depend only on their opponents' \emph{behavior}, not on other features of their source code.

\begin{definition}
Two agents $\texttt X$ and $\texttt Y$ are called \emph{behaviorally equivalent} if for every agent $\texttt Z$, $\mathsf{PA}$ proves $[X(Z)]\leftrightarrow[Y(Z)]$. A \emph{behavioral agent} is an agent $\texttt X$ such that for any pair of behaviorally equivalent agents $\texttt Y$ and $\texttt Z$, $\mathsf{PA}$ proves $[X(Y)]\leftrightarrow[X(Z)]$.
\end{definition}

\begin{theorem}
\label{behavioral}
Modal agents are behavioral.
\end{theorem}

\begin{proof}
Suppose that $\texttt X$ is a modal agent and $\texttt Y$ and $\texttt Z$ are behaviorally equivalent. Write $\texttt X_i$ for the lower-rank agents $\texttt X$ depends on, and $\varphi(p,q_1,\dotsc,q_n)$ for the modal formula corresponding to $\texttt X$. Then
\[
\mathsf{PA}\;\vdash\;[X(Y)]\;\leftrightarrow\;\varphi\big([Y(X)],[Y(X_1)],\dotsc,[Y(X_n)]\big)
\]
and
\[
\mathsf{PA}\;\vdash\;[X(Z)]\;\leftrightarrow\;\varphi\big([Z(X)],[Z(X_1)],\dotsc,[Z(X_n)]\big).
\]
But by the behavioral equivalence of $\texttt Y$ and $\texttt Z$ together with Lemma~\ref{lemma-modal-substitution}, the right-hand sides are equivalent, so $\mathsf{PA}\vdash[X(Y)]\leftrightarrow[X(Z)]$ as desired.
\end{proof}

\begin{corollary}
$\texttt{CliqueBot}$ is not a modal agent.
\end{corollary}

\begin{proof}
$\texttt{CliqueBot}$ cooperates with itself, but not with a syntactically different but logically (and therefore behaviorally) equivalent variant. Hence, $\texttt{CliqueBot}$ is not a behavioral agent, and by Theorem \ref{behavioral} it is not a modal agent.
\end{proof}

It feels a bit clunky, in some sense, for the definition of modal agents to include references to other, simpler modal agents. Could we not do just as well with a carefully constructed agent that makes no such outside calls (i.e. a modal agent of rank 0)?
\\
\\ Surprisingly, the answer is no: there is no modal agent of rank 0 that achieves mutual cooperation with \texttt{FairBot} and defects against \texttt{CooperateBot}. In particular:
\begin{theorem}
\label{thirdparties}
Any modal agent \texttt{X} of rank 0 such that $\mathsf{PA}\vdash[X(FB)=C]$ must also have $\mathsf{PA}\vdash[X(CB)=C]$.
\end{theorem}

\begin{proof}
Writing $\varphi(\cdot)$ for the modal formula defining \texttt{X}, by Lemma~\ref{lemma-modal-substitution} we see that \textsf{PA} proves $[X(CB) = C]\leftrightarrow\varphi(\top)$ and $[X(FB) = C]\leftrightarrow\varphi(\square[X(FB) = C])$; since by assumption it also proves $[X(FB) = C]$, we have $\mathsf{PA}\vdash\square[X(FB) = C]$ and hence $\mathsf{PA}\vdash\square[X(FB) = C]\leftrightarrow\top$, so with another application of Lemma~\ref{lemma-modal-substitution} we obtain $\mathsf{PA}\vdash\varphi(\square[X(FB) = C])\leftrightarrow\varphi(\top)$, whence $\mathsf{PA}\vdash[X(FB) = C]\leftrightarrow[X(CB) = C]$ and finally $\mathsf{PA}\vdash[X(CB) = C]$.
\end{proof}

\section{Obstacles to Optimality}
\label{optimality}

It is worthwhile to ask whether there is some meaningful sense of ``optimality'' for logical agents or modal agents in particular. For many natural definitions of optimality, this is impossible. For instance, there is no \texttt{X} such that for all $\texttt{Y}$, the utility achieved by \texttt{X} against $\texttt{Y}$ is the highest achieved by any \texttt{Z} against $\texttt{Y}$. (To see this, consider $\texttt{Y}$ defined so that $Y(Z)=C$ if and only if \texttt{Z}$\neq$\texttt{X}.) More generally, an agent can ``punish'' or ``reward'' other agents for any arbitrary feature of their code.

Could we hope that \texttt{PrudentBot} might at least be optimal among modal agents in some meaningful sense? As it happens, there are several impediments to optimality among modal agents, which together make it very difficult to formulate any nontrivial and non-vacuous definition of optimality.

Most directly, for any modal agents $\texttt{X}$ and $\texttt{Y}$, either their outputs are identical on all modal agents, or there exists a modal agent \texttt{Z} which cooperates against $\texttt{X}$ but defects against $\texttt{Y}$. (For an enlightening example, consider the agent \texttt{TrollBot} which cooperates with \texttt{X} if and only if $\textsf{PA}\vdash X(DB)=C$.) Thus any nontrivial and nonvacuous concept of optimality must be weaker than total dominance, and in particular it must accept that third parties could seek to ``punish'' an agent for succeeding in a particular matchup.\\
\\
Another issue is illustrated by the following agent:

\begin{algorithm}[H]
\label{JustBot}
 \SetKwInOut{Input}{Input}\SetKwInOut{Output}{Output}

 \Input{Source code of the agent \texttt{X}}
 \Output{$C$ or $D$}

 \eIf {\emph{\textsf{PA}} $\vdash$ X(FairBot)$=C$}{
  return $C$\;
  }{
  return $D$\;
  }
 \caption{\texttt{JustBot (JB)}}
\end{algorithm}
That is, \texttt{JustBot} cooperates with \texttt{X} if and only if \texttt{X} cooperates with \texttt{FairBot}. (Note that \texttt{JustBot} has different source code from \texttt{FairBot}; in particular, it can use a hard-cooded reference to \texttt{FairBot}'s code, where \texttt{FairBot} must use a quine.) Clearly, \texttt{JustBot} is exploitable by some algorithm (in particular, consider the non-modal algorithm which cooperates only with the corresponding \texttt{FairBot} and with nothing else), but since it is behaviorally equivalent to \texttt{FairBot}, by Theorem \ref{behavioral} it cannot be exploited by any modal agent.
\\
\\ Thirdly, agents can ``run out of deductive power'' and fail to make the right choices against opponents that use much higher formal systems in certain ways. Explicitly, consider the family of modal agents $\texttt{WaitFairBot}_K$, defined by $$[WaitFairBot_K(X)=C]\leftrightarrow ((\neg\Box^K\bot)\wedge \Box(\neg\Box^K\bot\to [X(WaitFairBot_K)=C])).$$ These are simply versions of FairBot operating in stronger formal systems, so we would like our modal agents to find mutual cooperation with them. As it turns out\footnote{We originally tried to include a proof of this claim, but it ballooned this section out of proportion. The interested reader should start by proving the folk theorem that Kripke frames for \textbf{GL} correspond to ordinal chains, and then show that for any modal agent \texttt{X} which defects against \texttt{DefectBot}, there is a number $n$ such that $X(DB)=D$ holds in every world above height $n$ within a Kripke frame, and then induct on agents and subformulas of agents to show that $X(WFB_K)=D$ for any $K>n$.}, any modal agent which defects against \texttt{DefectBot} will fail to achieve mutual cooperation with $\texttt{WaitFairBot}_K$ for all sufficiently large $K$.
\\
\\ Despite these reasons for pessimism, we have not actually ruled out the existence of a nontrivial and non-vacuous optimality criterion which corresponds to our philosophical intuitions about ``correct'' decisions. Additionally, there are a number of ways to depart only mildly from the modal framework (such as allowing quantifiers over agents), and these could invalidate some of the above obstacles.

\section{Philosophical Digressions}
\label{PhilosophicalConclusions}

One might ask (on a philosophical level) why we object to FairBot in the first place; isn't it a feature, not a bug, that this agent offers up cooperation even to agents that blindly trust it? We suggest that it is too tempting to anthropomorphize agents in this context, and that many problems which can be interpreted as playing a Prisoner's Dilemma against a CooperateBot are situations in which one would not hesitate to ``defect'' in real life without qualms.\\
\\ 
For instance, consider the following situation: You've come down with the common cold, and must decide whether to go out in public. If it were up to you, you'd stay at home and not infect anyone else. But it occurs to you that the cold virus has a ``choice'' as well: it could mutate and make you so sick that you'd have to go to the hospital, where it would have a small chance of causing a full outbreak! Fortunately, you know that cold viruses are highly unlikely to do this.\footnote{Incidentally, the reason that cold viruses do not pursue this strategy is that, throughout human history and most of the world today, making the host sicker gives the virus fewer, not more, chances to spread.} If you map out the payoffs, however, you find that you are in a Prisoner's Dilemma with the cold virus, and that it plays the part of a CooperateBot. Are you therefore inclined to ``cooperate'' and infect your friends in order to repay the cold virus for not making you sicker?\\
\\
The example is artificial and obviously facetious, but not entirely spurious. The world does not come with conveniently labeled ``agents''; entities on scales at least from viruses to governments show signs of goal-directed behavior. Given a sufficiently broad counterfactual, almost any of these could be defined as a CooperateBot on a suitable Prisoner's Dilemma. And most of us feel no compunction about optimizing our human lives without worrying about the flourishing of cold viruses.\footnote{Note that it would, in fact, be different if a virus were intelligent enough to predict the macroscopic behavior of their host and base their mutations on that! In such a case, one might well negotiate with the virus. Alternatively, if one's concern for the well-being of viruses reached a comparable level to one's concern for healthy friends, that would change the payoff matrix so that it was no longer a Prisoner's Dilemma. But both of these considerations are far more applicable to human beings than to viruses.}\\
\\
In a certain sense, PrudentBot is actually ``good enough'' among modal agents that one might expect to encounter: there are bound to be agents (CooperateBot and DefectBot) whose action fails to depend in any sense upon predictions of other agents' behavior, and other agents (FairBot, PrudentBot, etc) whose action depends meaningfully on such predictions. One should not expect to encounter a TrollBot or JustBot arising naturally! But it's worth pondering if this reasoning can be made formal in any elegant way.\\
\\
Do these results imply that sufficiently intelligent and rational agents will reach mutual cooperation in one-shot Prisoner's Dilemmas? In a word, no, not yet.\footnote{However, to borrow what Randall Munroe said about correlation and causation, this form of program equilibrium does waggle its eyebrows suggestively and gesture furtively (toward cooperation in the Prisoner's Dilemma) while mouthing `look over there'.} Many things about this setup are notably artificial, most prominently the perfectly reliable exchange of source code (and after that, the intractably long computations that might perhaps be necessary for even the finitary versions).\\
\\Nor does this have direct implications among human beings; our abilities to read each other psychologically, while occasionally quite impressive, bear only the slightest analogy to the extremely artificial setup of modal agents. Governments and corporations may be closer analogues to our agents (and indeed, game theory has been applied much more successfully on that scale than on the human scale), but the authors would not consider the application of these results to such organizations to be straightforward, either. The theorems herein are not a demonstration that a more advanced approach to decision theory (i.e. one which does not fail on what we consider to be common-sense problems) is practical, only a demonstration that it is possible.

\section{Open Problems}
\label{OpenProblems}
In particular, here are some open problems we have come across in this area:
\begin{itemize}
\item Is there a natural, nonvacuous, and nontrivial definition of optimality among modal agents?
\item Are there tractable ways of studying agents which can incorporate quantifiers as well as modal operators? For example, we might consider the non-modal agent \texttt{X} such that \texttt{X} cooperates with \texttt{Y} iff some formal system proves both that \texttt{Y} is unexploitable (given the consistency of some other formal system) and that \texttt{Y} cooperates with \texttt{X}.
\item What different dynamics arise when we consider the analogues of modal agents in more complicated games than the Prisoner's Dilemma? In particular, there are issues raised by games with more than one ``superrational equilibrium'', like the Coordination Game. The natural analogues of FairBot and PrudentBot transform any finite game between themselves into a coordination or bargaining game, but do not provide insight on how to resolve those sorts of conflicts.
\item What happens if we apply probabilistic reasoning rather than provability logic? As this allows for mixed strategies, it introduces all of the complexities of bargaining games, as well as new ones. 
\item What differs in games with more than two players; in particular, what might coordination and bargaining among coalitions look like? It is easier to imagine two agents with each others' source code agreeing to cooperate in a Prisoner's Dilemma than it is to imagine three agents with each others' source code agreeing on how to subdivide a fixed prize (which they lose if they do not have a majority agreeing on an acceptable split).
\end{itemize}

\section*{Acknowledgments}
This project was developed at a research workshop held by the Machine Intelligence Research Institute (MIRI) in Berkeley, California, in April 2013. The authors gratefully acknowledge MIRI's hospitality and support throughout the workshop.
\\
\\Patrick LaVictoire was partially supported by NSF Grant DMS-1201314 while working on this project.
\\
\\Thanks to everyone who has commented on various partial results and drafts, in particular Alex Altair, Stuart Armstrong, Andrew Critch, Wei Dai, Daniel Dewey, Gary Drescher, Kenny Easwaran, Cameron Freer, Bill Hibbard, Vladimir Nesov, Vladimir Slepnev, Jacob Steinhardt, Nisan Stiennon, Jessica Taylor, and Qiaochu Yuan. Further thanks to readers on the blog LessWrong for their comments on a preprint of this article.

\bibliography{DT}{}
\bibliographystyle{plain}

\end{document}